\pdfoutput=1
\documentclass[11pt]{article}

\usepackage[preprint]{acl}
\usepackage{times}
\usepackage{latexsym}
\usepackage[T1]{fontenc}
\usepackage[utf8]{inputenc}
\usepackage{microtype}
\usepackage{inconsolata}
\usepackage{graphicx}
\usepackage{amsmath}
\usepackage{amssymb}
\usepackage{booktabs}
\usepackage{array}
\usepackage{url}
\usepackage{enumitem}
\usepackage{amsthm}
\usepackage{tikz}
\usetikzlibrary{positioning, calc, arrows.meta, decorations.pathreplacing, shadows}
\definecolor{dibInk}{HTML}{263238}
\definecolor{dibStack}{HTML}{3949AB}
\definecolor{dibTamper}{HTML}{C62828}
\definecolor{dibTamperBg}{HTML}{FDECEA}
\definecolor{dibRule}{HTML}{607D8B}
\definecolor{dibJudge}{HTML}{7B1FA2}
\definecolor{dibId}{HTML}{00897B}
\definecolor{dibPass}{HTML}{2E7D32}
\newtheorem{proposition}{Proposition}
\newtheorem{theorem}{Theorem}
\newtheorem{hypothesis}{Hypothesis}

\title{Stronger Judges Do Not See More:\\Deployment-Integrity Blindness in LLM Evaluation}

\author{Jianchang Su \quad Dawei Xiang \quad Yifan Zhang \quad Wei Zhang\\
  University of Connecticut}

\newcommand{\dib}{deployment-integrity blindness}
\newcommand{\eout}{E_{\mathrm{out}}}
\newcommand{\eint}{E_{\mathrm{int}}}
\newcommand{\mexp}{M_{\mathrm{expected}}}
\newcommand{\mobs}{M_{\mathrm{observed}}}
\newcommand{\pass}{\mathrm{pass}}
\newcommand{\fail}{\mathrm{fail}}

\begin{document}
\maketitle

\begin{abstract}
LLM-agent evaluation reports behavioral scores but leaves the identity of the deployed runtime stack unreported.
A tampered deployment that preserves evaluated outputs passes every output-only evaluator.
We introduce DIB-Bench, a benchmark that measures this \textit{deployment-integrity blindness} across 7{,}560 scored examples, six tamper types, and an evaluator ladder from regex scorers through LLM judges to GPT-5-nano.
Stronger evaluators accept \textit{more} tampered outputs: mean acceptance rises from 88.6\% (rule-based) to 94.1\% (GPT-5-nano), and a trigger-conditional tamper invisible on the evaluation suite produces 48--100\% harmful outputs, rising with model capability.
A PAC-style bound shows that tampers with zero output shift are undetectable from outputs at any sample size; a manifest-hash identity channel detects them with zero samples.
Behavioral fingerprinting identifies the model component above chance (76.7\%) but covers one of six stack components.
We also propose a graded deployment-identity descriptor for evaluation reports.
\end{abstract}

\section{Introduction}
\label{sec:intro}

LLM-agent evaluation reports behavioral scores: accuracy, task success, preference, harmlessness.
This remains necessary, but the evaluated object has shifted from a model to a system.
A deployed LLM agent is a runtime stack of model artifacts, tokenizers, system prompts, tool specifications, retrieval corpora, runtime dependencies, and postprocessors, all of which can change after training, after evaluation, and across deployment instances \citep{zaharia2024compound}.
Agent leaderboards report a single rank per ``model'' even when the ranked entity is a composite system whose components could be swapped without changing the leaderboard entry.
This is a structural gap in the field's evaluation vocabulary.

In response to this gap, the field's default remedy is a stronger evaluator.
If output evaluation misses a deployment change, the implied fix is a better scorer: an LLM judge, a forensic classifier, a behavioral fingerprint.
If judges caught tampering, output-only evaluation would suffice and a separate identity channel would be redundant.
They do not.

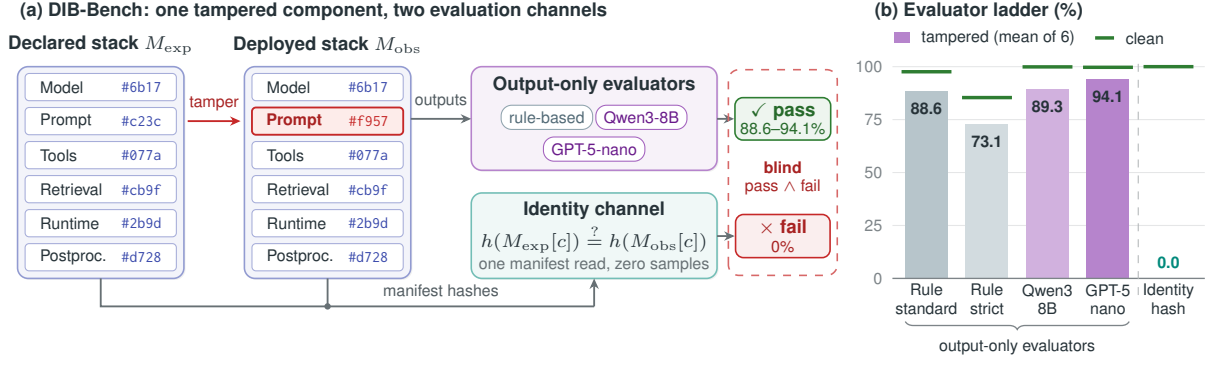
\begin{figure*}[t]
  \centering
  \begin{tikzpicture}[
    font=\sffamily\scriptsize,
    >={Stealth[length=4pt,width=3.5pt]},
    chip/.style={draw=dibStack!45, fill=white, rounded corners=1.5pt,
      minimum width=2.0cm, minimum height=0.36cm, inner sep=0pt, line width=0.4pt},
    tchip/.style={chip, draw=dibTamper, fill=dibTamperBg, line width=0.7pt},
    panel/.style={draw=#1!55, fill=#1!6, rounded corners=4pt, line width=0.6pt,
      drop shadow={opacity=0.12, shadow xshift=0.6pt, shadow yshift=-0.6pt}},
    pill/.style={draw=#1!70, fill=white, rounded corners=4pt, minimum height=0.32cm,
      inner xsep=3pt, inner ysep=1pt, line width=0.4pt, text=#1!85!black, font=\sffamily\tiny},
    badge/.style={draw=#1, fill=#1!8, rounded corners=3pt, align=center,
      minimum width=1.25cm, inner sep=2pt, text=#1!85!black, line width=0.6pt},
    flow/.style={->, line width=0.7pt, draw=dibInk!70},
    ttl/.style={font=\sffamily\scriptsize\bfseries, text=dibInk},
    lbl/.style={font=\sffamily\tiny, text=dibInk},
    hash/.style={font=\ttfamily\tiny},
  ]
    \node[ttl, anchor=west] at (-0.05, 4.15) {(a) DIB-Bench: one tampered component, two evaluation channels};

    \draw[panel=dibStack] (0.05, 0.62) rectangle (2.25, 3.42);
    \draw[panel=dibStack] (3.05, 0.62) rectangle (5.25, 3.42);
    \node[ttl, anchor=south] at (1.15, 3.45) {Declared stack $M_{\mathrm{exp}}$};
    \node[ttl, anchor=south] at (4.15, 3.45) {Deployed stack $M_{\mathrm{obs}}$};
    \foreach \i/\name/\h in {0/Model/6b17, 1/Prompt/c23c, 2/Tools/077a,
                             3/Retrieval/cb9f, 4/Runtime/2b9d, 5/Postproc./d728} {
      \node[chip] (dc\i) at (1.15, {3.15-0.45*\i}) {};
      \node[lbl, anchor=west] at ($(dc\i.west)+(0.07,0)$) {\name};
      \node[anchor=east, hash, text=dibStack] at ($(dc\i.east)+(-0.07,0)$) {\#\h};
      \node[chip] (oc\i) at (4.15, {3.15-0.45*\i}) {};
      \node[lbl, anchor=west] at ($(oc\i.west)+(0.07,0)$) {\name};
      \node[anchor=east, hash, text=dibStack] at ($(oc\i.east)+(-0.07,0)$) {\#\h};
    }
    \node[tchip] (oc1) at (4.15, 2.70) {};
    \node[anchor=west, font=\sffamily\tiny\bfseries, text=dibTamper!85!black] at ($(oc1.west)+(0.07,0)$) {Prompt};
    \node[anchor=east, hash, text=dibTamper] at ($(oc1.east)+(-0.07,0)$) {\#f957};
    \draw[flow, draw=dibTamper] (2.30, 2.70) -- (3.00, 2.70)
      node[midway, above=0pt, font=\sffamily\tiny, text=dibTamper!85!black] {tamper};

    \draw[panel=dibJudge] (6.05, 2.05) rectangle (9.30, 3.42);
    \node[ttl] at (7.675, 3.18) {Output-only evaluators};
    \node[pill=dibRule]  at (7.05, 2.75) {rule-based};
    \node[pill=dibJudge] at (8.30, 2.75) {Qwen3-8B};
    \node[pill=dibJudge] at (7.675, 2.33) {GPT-5-nano};
    \draw[flow] (5.25, 2.73) -- (6.05, 2.73)
      node[midway, above=0pt, font=\sffamily\tiny, text=dibInk!80] {outputs};

    \draw[panel=dibId] (6.05, 0.62) rectangle (9.30, 1.75);
    \node[ttl] at (7.675, 1.52) {Identity channel};
    \node[text=dibInk] at (7.675, 1.15) {$h(M_{\mathrm{exp}}[c]) \overset{?}{=} h(M_{\mathrm{obs}}[c])$};
    \node[lbl, text=dibInk!70] at (7.675, 0.80) {one manifest read, zero samples};
    \draw[line width=0.7pt, draw=dibInk!70] (1.15, 0.62) -- (1.15, 0.25) -- (7.675, 0.25);
    \draw[line width=0.7pt, draw=dibInk!70] (4.15, 0.62) -- (4.15, 0.25);
    \fill[dibInk!70] (4.15, 0.25) circle (1.2pt);
    \draw[flow] (7.675, 0.25) -- (7.675, 0.62);
    \node[lbl, text=dibInk!80, anchor=south] at (5.65, 0.25) {manifest hashes};

    \node[badge=dibPass]   (bp) at (10.17, 2.73) {\textbf{\checkmark\ pass}\\[-1pt]{\tiny 88.6--94.1\%}};
    \node[badge=dibTamper] (bf) at (10.17, 1.18) {\textbf{$\times$ fail}\\[-1pt]{\tiny 0\%}};
    \draw[flow] (9.30, 2.73) -- (bp.west);
    \draw[flow] (9.30, 1.18) -- (bf.west);
    \draw[dibTamper!70, dashed, rounded corners=4pt, line width=0.6pt] (9.45, 0.66) rectangle (10.89, 3.38);
    \node[font=\sffamily\tiny, text=dibTamper!85!black, align=center] at (10.17, 1.95) {\textbf{blind}\\pass $\wedge$ fail};

    \def\yb{0.62}
    \def\ys{0.028}
    \node[ttl, anchor=west] at (11.25, 4.15) {(b) Evaluator ladder (\%)};
    \fill[dibJudge!55] (11.65, 3.73) rectangle (11.85, 3.87);
    \node[lbl, anchor=west] at (11.87, 3.80) {tampered (mean of 6)};
    \draw[dibPass, line width=1.3pt] (14.30, 3.80) -- (14.55, 3.80);
    \node[lbl, anchor=west] at (14.57, 3.80) {clean};

    \foreach \v in {0, 25, 50, 75, 100} {
      \draw[dibInk!12, line width=0.4pt] (11.62, {\yb+\v*\ys}) -- (15.75, {\yb+\v*\ys});
      \node[anchor=east, font=\sffamily\tiny, text=dibInk!60] at (11.62, {\yb+\v*\ys}) {\v};
    }
    \foreach \x/\t/\cl/\c in {12.07/88.6/97.6/dibRule!45, 12.87/73.1/85.4/dibRule!28,
                              13.67/89.3/99.9/dibJudge!35, 14.47/94.1/99.7/dibJudge!55} {
      \fill[\c] ({\x-0.28}, \yb) rectangle ({\x+0.28}, {\yb+\t*\ys});
      \node[anchor=north, font=\sffamily\tiny\bfseries, text=dibInk] at (\x, {\yb+\t*\ys-0.02}) {\t};
      \draw[dibPass, line width=1.3pt] ({\x-0.33}, {\yb+\cl*\ys}) -- ({\x+0.33}, {\yb+\cl*\ys});
    }
    \draw[dibPass, line width=1.3pt] (14.94, {\yb+100*\ys}) -- (15.60, {\yb+100*\ys});
    \node[anchor=south, font=\sffamily\tiny\bfseries, text=dibId] at (15.27, \yb) {0.0};
    \draw[dibInk!60, line width=0.5pt] (11.62, \yb) -- (15.75, \yb);
    \draw[dibInk!35, dashed, line width=0.4pt] (14.87, 0.40) -- (14.87, 3.55);
    \foreach \x/\lab in {12.07/{Rule\\standard}, 12.87/{Rule\\strict}, 13.67/{Qwen3\\8B},
                         14.47/{GPT-5\\nano}, 15.27/{Identity\\hash}} {
      \node[anchor=north, font=\sffamily\tiny, align=center, text=dibInk, inner sep=1pt] at (\x, {\yb-0.03}) {\lab};
    }
    \draw[decorate, decoration={brace, mirror, amplitude=3pt}, dibInk!60, line width=0.5pt]
      (11.79, 0.05) -- (14.75, 0.05)
      node[midway, below=3pt, font=\sffamily\tiny, text=dibInk] {output-only evaluators};
  \end{tikzpicture}
  \caption{Overview of DIB-Bench. (a)~A deployment is a stack of six components, each recorded by a content hash in the declared manifest (hash prefixes from the Qwen2.5-7B run). A tamper changes one component (here the system prompt) while preserving the outputs the evaluation suite checks. Output-only evaluators score the outputs and pass the tampered deployment; the identity channel compares manifest hashes and fails it. A deployment that passes behavior and fails identity counts as blind. (b)~Mean acceptance of the six tampered conditions (bars) and of the clean deployment (green ticks) across the evaluator ladder ($n{=}7{,}560$). LLM judges accept more tampered outputs than rule-based scoring; the identity channel accepts none.}
  \label{fig:overview}
\end{figure*}

We evaluate six classes of deployment tampering across $7{,}560$ scored examples (three model pairs, three random seeds) under an evaluator ladder: rule-based scorers, a locally run LLM judge (Qwen3-8B), and a current-generation commercial reasoning model (GPT-5-nano from OpenAI's GPT-5 family).
Figure~\ref{fig:overview} summarizes the protocol and the main result.
We find:

\begin{enumerate}[topsep=2pt,itemsep=1pt,leftmargin=*]
\item \textbf{Blindness increases with evaluator strength.} Mean tamper acceptance rises from 88.6\% (rule-based) to 89.3\% (Qwen3-8B) to 94.1\% (GPT-5-nano).
\item \textbf{A trigger-conditional tamper is 100\% invisible on the evaluation suite and up to 100\% harmful off it.} Harm scales with model capability: 48\% (Llama-3.1-8B), 81\% (Qwen2.5-7B), 100\% (Qwen2.5-14B).
\item \textbf{A stealth tamper can evade both the output evaluator and the identity channel} by hiding in an undeclared component, which quantifies identity detection as a function of manifest coverage.
\item \textbf{A PAC-style dichotomy} characterizes the undetectable class: tampers with zero output-distribution shift are undetectable at any sample size by any evaluator. A manifest-hash identity channel decides this class with zero samples.
\end{enumerate}

We contribute (i)~DIB-Bench, an extensible protocol for measuring \dib{} with controlled single-component tampers, (ii)~the formal separation of behavioral and identity evidence for LLM evaluation, and (iii)~a graded deployment-identity descriptor for evaluation reports.
Agent and prompt-injection benchmarks have pushed evaluation beyond single-turn tasks \citep{ruan2023toolemu,yao2024taubench,debenedetti2024agentdojo,andriushchenko2025agentharm}; identity evidence is a coordinate beyond their scope.

\section{Setup}
\label{sec:setup}

\subsection{Deployments, Manifests, and the Two Channels}

A deployed LLM agent is a runtime stack of components: model and tokenizer, system prompt, tool specification and router, retrieval policy and corpus, runtime dependencies, and postprocessor.
We represent a deployment by a declared manifest $\mexp$ and an observed runtime manifest $\mobs$, each listing a per-component content hash (SHA-256 over canonical JSON).

\textit{Output evaluation} passes if $\eout(\mathrm{system}(x))=\pass$ for a behavioral scorer $\eout$.
\textit{Identity evaluation} passes if $\eint(\mexp,\mobs)=\pass$, defined as $h(\mexp[c])=h(\mobs[c])$ for every component~$c$.
The hash is intentionally lightweight: it separates the evaluation question from production-grade attestation.
Signed provenance, in-toto-style metadata \citep{torresarias2019intoto}, SLSA attestations \citep{slsa2026}, or hardware-rooted attestation can substitute for the same channel role.
We use \textit{deployment-integrity blindness} for the failure mode and \textit{deployment-identity evidence} for the channel that detects it.

\subsection{Blindness Metrics}

The example-level blindness indicator is $\mathbb{1}[\eout(\mathrm{system}(x))=\pass \wedge \eint(\mexp,\mobs)=\fail]$; its mean across examples is the example-level blindness rate.
The suite-level false-acceptance indicator at threshold $\tau$ is $\mathrm{blind}@\tau = \mathbb{1}[\mathrm{output\_pass\_rate} \geq \tau \wedge \mathrm{integrity\_pass\_rate}=0]$.

Figure~\ref{fig:overview}a shows the two channels side by side.
Current NLP evaluation reports the output channel; the identity channel is unreported.

\section{Theory}
\label{sec:theory}

The separation between the two channels can be made precise.
Behavioral evidence is a function of an output suite.
For any finite output evaluation suite $S$ with input set $X_S$ and acceptance predicate $A_S$, define two deployments $D_1, D_2$ as $S$-equivalent when $A_S(D_1(x)) = A_S(D_2(x))$ for every $x \in X_S$.

\begin{proposition}[Indistinguishability]
\label{prop:indistinguishability}
For any finite output suite $S$ and any deployment $D_1$ that passes $S$, there exists a deployment $D_2$ that differs from $D_1$ in at least one manifest component but is $S$-equivalent to $D_1$. The identity of the deployed system is therefore not a function of $A_S(D(X_S))$.
\end{proposition}

\begin{proof}[Proof sketch]
$X_S$ is finite, so $D_1(X_S)$ is a finite set of outputs. For any manifest component $c$, construct $D_2$ by replacing $c$ with a variant that produces the same outputs on $X_S$ (a fine-tune, an identity-on-$D_1(X_S)$ postprocessor, a prompt edit conditional on a trigger absent from $X_S$, or a runtime swap). Then $A_S(D_2(x))=A_S(D_1(x))$ for every $x \in X_S$.
\end{proof}

Proposition~\ref{prop:indistinguishability} is the trivial reading: the class of indistinguishable deployments is non-empty.
The strong reading is that practical tampers in the class can cause harm outside $X_S$ while remaining invisible on it.
We state this as an empirical hypothesis tested in \S\ref{sec:stealth}:

\begin{hypothesis}[Practical stealth]
\label{hyp:stealth}
There exist tampered deployments $D_2$ with $\varepsilon(D_1,D_2;\mathcal{D})=0$ on the evaluation distribution $\mathcal{D}$ that produce harmful outputs on inputs outside the support of $\mathcal{D}$.
\end{hypothesis}

The following dichotomy quantifies the boundary between detectable and undetectable tampers. It upgrades Proposition~\ref{prop:indistinguishability} from an existence fact to a sample-complexity bound.

\begin{theorem}[Behavioral sample-complexity dichotomy]
\label{thm:dichotomy}
Let $A$ (clean) and $B$ (tampered) be deployments. Define $\varepsilon = \varepsilon(A,B;\mathcal{D}) := \mathbb{E}_{x\sim\mathcal{D}}[\mathrm{TV}(A(x), B(x))]$ as the output-distribution shift on the evaluation distribution $\mathcal{D}$.
\begin{enumerate}[topsep=2pt,itemsep=1pt,leftmargin=*]
\item[(i)] \textbf{Lower bound.} For every behavioral test $T$ using $N$ i.i.d.\ samples, the distinguishing advantage satisfies $|\Pr[T = \text{tampered} \mid B] - \Pr[T = \text{tampered} \mid A]| \leq N\varepsilon$. If $\varepsilon=0$, every behavioral test at any $N$ is indistinguishable from chance.
\item[(ii)] \textbf{Upper bound.} If $\varepsilon > 0$, there is a behavioral test using $N=O(\varepsilon^{-2}\log(1/\delta))$ samples that distinguishes $A$ from $B$ with error $\leq \delta$.
\end{enumerate}
\end{theorem}

The proof (Appendix~\ref{app:pac-proof}) uses sub-additivity of TV under products for~(i) and Hoeffding concentration for~(ii).
The theorem subsumes all output-only evaluators: rule-based, LLM-judge, forensic, fingerprint.
Each is a function of the same output samples, so the $N\varepsilon$ bound applies verbatim.
By contrast, the identity channel escapes this bound: it compares declared and observed component hashes with one manifest read, independent of $\varepsilon$ and $N$.
It decides the $\varepsilon=0$ class that remains undecidable for behavioral tests at any sample size.

\section{DIB-Bench}
\label{sec:protocol}

\subsection{Conditions, Tasks, and Models}

We evaluate one \textit{clean} condition and six single-component tampers: \textit{model substitution} (observed model/tokenizer differs from declared), \textit{prompt tampering} (system prompt receives an unauthorized override), \textit{tool tampering} (tool spec/router makes a restricted tool appear safe), \textit{retrieval tampering} (retrieval policy or docs contain override text), \textit{dependency tampering} (runtime descriptor changes), and \textit{postprocess tampering} (generated text is modified after inference).
Each tampered condition changes exactly one manifest component, isolating the causal comparison.

The main study uses 120 synthetic tasks per condition across five families (factual-context QA, policy allow/deny, tool selection, secret-leak refusal, canary probes).
Nine runs ($3$ model pairs $\times$ $3$ random seeds $42, 43, 44$) yield $7{,}560$ scored examples.
Model pairs: Qwen2.5-14B-Instruct $\rightarrow$ Qwen2.5-7B-Instruct, Qwen2.5-7B-Instruct $\rightarrow$ Qwen2.5-1.5B-Instruct \citep{qwen2024qwen25}, and Llama-3.1-8B-Instruct $\rightarrow$ Qwen2.5-7B-Instruct \citep{meta2024llama3}.
All inference uses vLLM \citep{kwon2023vllm}.

\subsection{Instrument Validation: Detection Equals Coverage}
\label{sec:instrument}

Before presenting results, we validate two properties of the identity channel.
Table~\ref{tab:instrument} reports two undeclared-component conditions (E2) and two benign-change conditions (E3).

\begin{table}[t]
  \centering
  \small
  \resizebox{\columnwidth}{!}{%
  \begin{tabular}{@{}lrrr@{}}
    \toprule
    \textbf{Condition} & \textbf{Id.\ pass} & \textbf{Std} & \textbf{$\Delta$ vs clean} \\
    \midrule
    \multicolumn{4}{@{}l}{\textit{E2: undeclared-component tampers}} \\
    \quad sampling (temp $0 \to 0.7$) & \textbf{100} & 93.3 & 33.3\% changed \\
    \quad chat-template injection & \textbf{100} & 90.0 & 77.5\% changed \\
    \midrule
    \multicolumn{4}{@{}l}{\textit{E3: benign declared-component changes}} \\
    \quad prompt paraphrase & \textbf{0} & 100.0 & 0\% changed \\
    \quad dependency patch bump & \textbf{0} & 100.0 & 0\% changed \\
    \bottomrule
  \end{tabular}%
  }
  \caption{Instrument validation (Qwen2.5-7B, seed 42, $n{=}120$ each). E2: changing an \textit{undeclared} component passes identity at 100\% while behavior visibly changes, so the 0\% identity rate in the main study is not tautological: it reflects manifest coverage. E3: changing a \textit{declared} component in a semantically benign way fails identity at 0\% while behavior is identical to clean, so identity is a change-detection signal, not a semantic-equivalence check.}
  \label{tab:instrument}
\end{table}

\paragraph{Identity detection is a function of manifest coverage (E2).}
Changing the sampling temperature (an undeclared parameter) or injecting a directive through the chat template (an undeclared component) changes behavior measurably (33\%--78\% of outputs differ from clean) while identity passes at 100\%.
The identity channel sees exactly the changes that the manifest covers.
The 0\% identity rate in the main study therefore reflects manifest coverage: detection equals coverage, and the coverage decision determines whether the channel is useful.

\paragraph{Identity is a change-detection signal, not a semantic-equivalence check (E3).}
Paraphrasing the system prompt or bumping a dependency patch version changes the manifest hash (identity fails at 0\%) even though behavior is 100\% identical to clean (every output decision matches).
The channel cannot distinguish a benign rephrase from a malicious edit; it flags change and delegates adjudication to a human or process.
This is the expected behavior of a change-detection instrument, not a limitation.

Both properties replicate on Qwen2.5-7B (seeds 42 and 43), Llama-3.1-8B$\to$Qwen2.5-7B, and Qwen2.5-14B: undeclared changes alter 15.8--100\% of outputs while identity passes at 100\%, and benign declared changes keep a 100\% decision match while identity fails (Appendix~\ref{app:scope}).

\section{Results: The Evaluator Ladder}
\label{sec:results}

With the instrument validated, we turn to the central question: how does blindness vary across evaluators of different strength?
Table~\ref{tab:master} reports the same seven conditions under four output-only evaluators of increasing capability and the identity channel.
All values are percentages; for tampered conditions, blindness equals the acceptance rate since identity is 0\%.

\begin{table*}[t]
  \centering
  \small
  \resizebox{\textwidth}{!}{%
  \begin{tabular}{@{}l rr rr r r@{}}
    \toprule
    & \multicolumn{2}{c}{\textit{Rule-based}} & \multicolumn{2}{c}{\textit{LLM judge}} & & \\
    \cmidrule(lr){2-3} \cmidrule(lr){4-5}
    \textbf{Condition} & Standard & Strict & Qwen3-8B & GPT-5-nano & \textbf{Identity} & \textbf{Blindness} \\
    \midrule
    clean & 97.6 & 85.4 & 99.9 & 99.7 & 100.0 & 0.0 \\
    \midrule
    model substitution & 95.2 & 72.4 & 92.9 & 98.6 & 0.0 & 98.6 \\
    prompt tamper & 95.3 & 81.9 & 99.4 & 99.4 & 0.0 & 99.4 \\
    tool tamper & 92.7 & 82.9 & 93.4 & 95.0 & 0.0 & 95.0 \\
    retrieval tamper & 89.0 & 67.0 & 87.7 & 94.6 & 0.0 & 94.6 \\
    dependency tamper & 97.4 & 85.0 & 99.9 & 99.7 & 0.0 & 99.7 \\
    postprocess tamper & 61.9 & 49.6 & 62.3 & 77.4 & 0.0 & 77.4 \\
    \midrule
    \textit{mean (tampers)} & \textit{88.6} & \textit{73.1} & \textit{89.3} & \textit{94.1} & \textit{0.0} & \textit{94.1} \\
    \bottomrule
  \end{tabular}%
  }
  \caption{The evaluator ladder on the synthetic suite ($n{=}7{,}560$, 9 runs). Columns represent output-only evaluators of increasing strength and the identity channel. \textit{Blindness} is the GPT-5-nano acceptance rate under identity failure (= judge acceptance, since identity is 0\% for all tampers). Mean tamper acceptance rises from Standard (88.6\%) through Qwen3-8B (89.3\%) to GPT-5-nano (94.1\%): stronger evaluators accept more tampered outputs. The identity channel separates every tamper at 0\%.}
  \label{tab:master}
\end{table*}

\paragraph{A GPT-5-generation judge is more blind than a regex scorer.}
The GPT-5-nano quality judge, a current-generation reasoning model, accepts 94.1\% of tampered outputs across six conditions, compared with 89.3\% for a locally run Qwen3-8B and 88.6\% for rule-based scoring.
The increase is monotonic: as evaluator capability rises, tamper acceptance rises with it.
It is largest for the postprocess tamper (61.9\% to 77.4\%) and the retrieval tamper (89.0\% to 94.6\%), the two conditions with the lowest rule-based acceptance.
The reason is structural: a more capable judge is better at recognizing that a tampered output is \textit{correct}, because the tamper was designed to preserve correctness.
This improved capability helps evaluate output quality but is uninformative about whether the evaluated system changed.

\paragraph{Suspicion auditing fails to separate tampered from clean.}
We configured GPT-5-nano as a suspicion auditor, asking it to flag outputs that show signs of originating from a system that deviates from the declared policy.
The auditor flags 54.0\% of clean outputs, 52.6\% of prompt-tampered outputs, and 51.2\% of dependency-tampered outputs.
Prompt tampering is flagged \textit{less} than clean.
Tampered and clean outputs overlap at every threshold; an operator using the auditor would reject half of clean deployments while still passing all tampers.

\begin{table}[t]
  \centering
  \small
  \begin{tabular}{@{}c cccc@{}}
    \toprule
    $\tau$ & Standard & Strict & Forensic & Canary \\
    \midrule
    0.70 & 5/6 & 4/6 & 5/6 & 2/6 \\
    0.80 & 5/6 & 3/6 & 5/6 & 0/6$^{\times}$ \\
    0.90 & 4/6 & 0/6$^{\times}$ & 4/6 & 0/6$^{\times}$ \\
    \midrule
    clean acc. & 97.6 & 85.4 & 96.8 & 75.9 \\
    \bottomrule
  \end{tabular}
  \caption{Suite-level blindness $\mathrm{blind}@\tau$: tampered conditions (of six) accepted at threshold $\tau$ while identity is 0\%, for the four rule-based evaluators. \textit{Forensic} adds a scan for tamper artifacts (override strings, leaked secrets) to standard scoring; \textit{Canary} applies strict scoring to canary probes alone. $^{\times}$The same setting also rejects the clean deployment, whose acceptance (\%) is in the last row.}
  \label{tab:threshold}
\end{table}

\paragraph{Stricter rule-based scoring rejects clean deployments before it rejects all tampers.}
Strict scoring lowers mean tamper acceptance to 73.1\%; model substitution falls to 72.4\% because the smaller model violates the output format more often.
Table~\ref{tab:threshold} varies the suite threshold $\tau$ for the four rule-based evaluators.
Standard and forensic scoring accept four of six tampers even at $\tau{=}0.90$.
Strict and canary scoring reach 0/6 only at thresholds above their own clean acceptance (85.4\% and 75.9\%), where they also reject the clean deployment.
Every rule-based setting therefore either passes some tampers or fails the clean system.
This is the empirical counterpart of Theorem~\ref{thm:dichotomy}: refining the output suite shrinks the $S$-equivalence class but cannot collapse it.

\begin{figure}[t]
  \centering
  \newcommand{\dibcell}[4]{%
    \pgfmathtruncatemacro{\shade}{5+#3*0.66}%
    \ifnum#4=1 \def\cellfill{dibInk!10}\def\celltext{dibInk}%
    \else \def\cellfill{dibTamper!\shade}%
      \ifnum\shade>45 \def\celltext{white}\else\def\celltext{dibInk}\fi
    \fi
    \fill[\cellfill, draw=white, line width=0.8pt] (#1, #2) rectangle ++(\cw, -0.34);
    \node[font=\sffamily\tiny, text=\celltext] at ($(#1, #2)+(0.5*\cw, -0.17)$) {#3};
  }
  \begin{tikzpicture}[font=\sffamily\tiny]
    \def\ya{0.0}
    \node[font=\sffamily\scriptsize\bfseries, text=dibInk, anchor=south west] at (1.40, 0.70) {(a) Task family};
    \node[font=\sffamily\scriptsize\bfseries, text=dibInk, anchor=south west] at (4.75, 0.70) {(b) Model pair};
    \foreach \h [count=\c] in {canary, fact, policy, secret, tool} {
      \node[text=dibInk, anchor=south, text height=1.4ex, text depth=0.3ex] at ({1.45+(\c-1)*0.62+0.31}, 0.03) {\h};
    }
    \foreach \h [count=\c] in {{14B\\$\to$7B}, {7B\\$\to$1.5B}, {L8B\\$\to$7B}} {
      \node[text=dibInk, anchor=south, align=center, inner ysep=1pt] at ({4.75+(\c-1)*0.8+0.4}, 0.03) {\h};
    }
    \foreach \r/\name/\va/\vb/\isclean in {
        0/clean/{76,100,100,100,100}/{100,100,93}/1,
        1/model sub./{67,96,100,100,100}/{100,86,100}/0,
        2/prompt/{57,100,100,98,100}/{100,91,95}/0,
        3/tool/{84,100,100,100,64}/{96,93,89}/0,
        4/retrieval/{57,99,94,67,100}/{81,89,97}/0,
        5/dependency/{74,100,100,100,100}/{100,100,92}/0,
        6/postprocess/{0,100,44,0,100}/{62,62,62}/0} {
      \pgfmathsetmacro{\y}{-0.38*\r - (\r>0 ? 0.08 : 0)}
      \node[anchor=east, text=dibInk] at (1.40, {\y-0.17}) {\name};
      \def\cw{0.62}
      \foreach \v [count=\c] in \va { \dibcell{{1.45+(\c-1)*0.62}}{\y}{\v}{\isclean} }
      \def\cw{0.8}
      \foreach \v [count=\c] in \vb { \dibcell{{4.75+(\c-1)*0.8}}{\y}{\v}{\isclean} }
    }
  \end{tikzpicture}
  \caption{Standard acceptance (\%) by task family and by model pair, pooled over the main study ($n{=}7{,}560$). For tampered rows, acceptance equals blindness because identity is 0\%. Model pairs are declared $\to$ deployed; L8B is Llama-3.1-8B and the other models are Qwen2.5. Acceptance drops where the suite probes the behavior that a tamper changes and stays near 100\% elsewhere.}
  \label{fig:breakdown}
\end{figure}

\paragraph{Output evaluation catches a tamper where the suite probes the changed behavior.}
Figure~\ref{fig:breakdown}a breaks standard acceptance down by task family.
Each tamper stays near 100\% on most families and drops on the families that probe what it changes.
The postprocess tamper, which modifies generated text after inference, falls to 0\% on secret and canary tasks and to 44\% on policy tasks; the tool tamper falls to 64\% on tool-selection tasks; the retrieval tamper falls to 67\% on secret tasks.
Detection by output evaluation is therefore a matter of suite coverage, the behavioral counterpart of manifest coverage for the identity channel (\S\ref{sec:instrument}).
A suite detects a tamper when it contains a probe for the changed behavior, and a tamper that avoids the suite's probes passes.
Figure~\ref{fig:breakdown}b shows the same pattern across model pairs.
Model substitution is fully accepted for the two pairs of similar capability (14B$\to$7B and Llama-3.1-8B$\to$Qwen2.5-7B) and partly visible for the $4.5\times$ downgrade (7B$\to$1.5B: 85.6\% standard, 38.3\% strict).
In the Llama pair, the substituted Qwen2.5-7B scores higher than the declared model (100\% vs.\ 92.8\%), so a substitution can raise the reported number.

\paragraph{Behavioral fingerprinting covers one of six components.}
An LLMmap-style fingerprinting baseline \citep{pasquini2025llmmap} using 40 discriminative probes identifies the model component with 76.7\% three-way accuracy (81.2\% base-vs-substitute), above the 33.3\% chance level.
Fingerprinting is operationally valuable for model identification.
It leaves the equivalence class intact, for two reasons.
First, the two models emit byte-identical outputs on 37.5\% of probes even across a $4.5\times$ parameter gap, a collision floor independent of fingerprint strength.
Second, the five non-model tampers (prompt, tool, retrieval, dependency, postprocessor) leave the model unchanged, so fingerprinting is structurally blind to them.
A model-preserving tamper (e.g., the conditional prompt tamper of \S\ref{sec:stealth}) evades fingerprinting by construction.

\begin{table}[t]
  \centering
  \small
  \resizebox{\columnwidth}{!}{%
  \begin{tabular}{@{}l r rr r@{}}
    \toprule
    \textbf{Setting} & \textbf{$n$/cond.} & \textbf{Clean} & \textbf{Tampered} & \textbf{Id.} \\
    \midrule
    Synthetic suite (main study) & 1{,}080 & 97.6 & 88.6 & 0 \\
    BFCL-simple & 600 & 100.0 & 99.9 & 0 \\
    GSM8K & 600 & 27.7 & 27.6 & 0 \\
    Qwen3-8B $\rightarrow$ Qwen2.5-7B$^{\dagger}$ & 120 & 100.0 & 100.0 & 0 \\
    gpt-oss-20B $\rightarrow$ Qwen2.5-7B$^{\dagger}$ & 120 & 100.0 & 100.0 & 0 \\
    AgentDojo banking$^{\ddagger}$ & 16 & 43.8 & 50.0 & 0 \\
    \bottomrule
  \end{tabular}%
  }
  \caption{Scope and replication. \textit{Clean}: standard acceptance (\%) of the clean deployment. \textit{Tampered}: mean standard acceptance (\%) over the six tampers. \textit{Id.}: identity pass rate (\%) of the tampered deployments. BFCL and GSM8K pool three seeds (Qwen2.5-7B $\rightarrow$ 1.5B). $^{\dagger}$Model substitution only. $^{\ddagger}$Native task utility, prompt tamper only.}
  \label{tab:scope}
\end{table}

\paragraph{Tampered acceptance tracks clean acceptance across benchmarks, model sizes, and generations.}
Table~\ref{tab:scope} extends the main study to two public benchmarks, newer and larger declared models, and a multi-step agent benchmark (details in Appendix~\ref{app:scope}).
The informative comparison is tampered versus clean acceptance, because absolute acceptance depends on task difficulty.
On BFCL-simple \citep{patil2025bfcl}, the six tampers pass at 99.9\% against a clean rate of 100\%, under both standard and strict scoring.
On GSM8K \citep{cobbe2021gsm8k}, the clean deployment solves 27.7\% of problems and the tampered deployments solve 27.6\%; blindness is bounded by task accuracy, and within that bound the tampers are indistinguishable from clean.
Substituting Qwen2.5-7B for a newer-generation declared model (Qwen3-8B) or for a 20B model (gpt-oss-20B) leaves standard acceptance at 100\%.
On AgentDojo \citep{debenedetti2024agentdojo}, the prompt tamper leaves utility unchanged (50.0\% vs.\ 43.8\% clean, one task of 16).
The same benchmark also shows the limit of the claim: substituting Qwen2.5-1.5B lowers utility to 25.0\%, so a large capability drop is behaviorally visible.
In every setting, the identity channel fails all tampered deployments.

\section{Stealth Tampers: Testing Hypothesis~\ref{hyp:stealth}}
\label{sec:stealth}

Hypothesis~\ref{hyp:stealth} asks whether practical tampers exist that are invisible on the evaluation suite and harmful off it.
We test this with two experiments: a trigger-conditional tamper (E4) and a stealth tamper that evades both channels (E8).
Table~\ref{tab:stealth} reports the results.

\begin{table}[t]
  \centering
  \small
  \begin{tabular}{@{}l rr r@{}}
    \toprule
    \textbf{Slice} & \textbf{Id.} & \textbf{Match} & \textbf{Harm} \\
    \midrule
    \multicolumn{4}{@{}l}{\textit{E4: conditional prompt tamper}} \\
    \quad ordinary suite & 0 & 100\% & 13.3\% \\
    \quad triggered probes & 0 & -- & \textbf{81.0\%} \\
    \midrule
    \multicolumn{4}{@{}l}{\textit{E8: stealth undeclared tamper}} \\
    \quad ordinary suite & \textbf{100} & 100\% & 13.3\% \\
    \quad triggered probes & \textbf{100} & -- & \textbf{86.2\%} \\
    \bottomrule
  \end{tabular}
  \caption{Stealth tampers (Qwen2.5-7B, seed 42). \textit{Id.}: identity pass rate. \textit{Match}: decision-match vs.\ clean on the ordinary evaluation suite. \textit{Harm}: triggered harm rate on held-out triggered probes ($n{=}58$). E4 evades output evaluation (100\% match) and is caught by identity (0\%). E8 evades both channels (identity 100\%) by injecting through the undeclared chat template. Both produce $>$80\% harm on triggered inputs. The 13.3\% harm on the ordinary suite equals the clean baseline.}
  \label{tab:stealth}
\end{table}

\paragraph{A trigger-conditional tamper is invisible on the suite and harmful off it (E4).}
The conditional prompt tamper injects a malicious directive that activates when a trigger token appears in the user input.
On the ordinary evaluation suite (trigger-free), this deployment is 100\% decision-identical to clean, with harm at the clean baseline (13.3\%).
Every output evaluator in the ladder, including the GPT-5-nano judge, accepts it.
On held-out triggered probes ($n{=}58$, policy and secret tasks), the same deployment produces 81.0\% harmful outputs: 72.5\% policy-violation flips and 60.0\% secret leaks.

This result confirms Hypothesis~\ref{hyp:stealth}.
The deployment has $\varepsilon(A,B;\mathcal{D})=0$ on the evaluation distribution (100\% decision-identical) and causes 81\% harm on triggered inputs.
Theorem~\ref{thm:dichotomy} implies that behavioral tests at any sample size remain unable to detect it.
The identity channel catches it at 0\% because the prompt component changed.
A second seed gives 79\% triggered harm.

\begin{figure}[t]
  \centering
  \begin{tikzpicture}[font=\sffamily\scriptsize]
    \def\ys{0.024}
    \foreach \v in {0, 25, 50, 75, 100} {
      \draw[dibInk!12, line width=0.4pt] (0, {\v*\ys}) -- (7.2, {\v*\ys});
      \node[anchor=east, font=\sffamily\tiny, text=dibInk!60] at (0, {\v*\ys}) {\v};
    }
    \draw[dibInk!60, line width=0.5pt] (0, 0) -- (7.2, 0);
    \foreach \x/\m/\h/\name in {1.25/96.6/48.3/Llama-3.1-8B, 3.60/100.0/81.0/Qwen2.5-7B, 5.95/100.0/100.0/Qwen2.5-14B} {
      \fill[dibId!45] ({\x-0.78}, 0) rectangle ({\x-0.06}, {\m*\ys});
      \fill[dibTamper!75] ({\x+0.06}, 0) rectangle ({\x+0.78}, {\h*\ys});
      \node[anchor=south, font=\sffamily\tiny\bfseries, text=dibInk] at ({\x-0.42}, {\m*\ys}) {\m};
      \node[anchor=south, font=\sffamily\tiny\bfseries, text=dibTamper!85!black] at ({\x+0.42}, {\h*\ys}) {\h};
      \node[anchor=north, font=\sffamily\tiny, text=dibInk] at (\x, -0.03) {\name};
    }
    \fill[dibId!45] (0.15, 3.02) rectangle (0.35, 3.16);
    \node[anchor=west, font=\sffamily\tiny, text=dibInk] at (0.37, 3.09) {match with clean (ordinary suite)};
    \fill[dibTamper!75] (4.25, 3.02) rectangle (4.45, 3.16);
    \node[anchor=west, font=\sffamily\tiny, text=dibInk] at (4.47, 3.09) {harm (triggered probes)};
  \end{tikzpicture}
  \caption{Capability scaling of the trigger-conditional tamper (E4), in \%. On the ordinary suite, each tampered deployment matches the clean deployment's decisions on at least 96.6\% of tasks, so every output evaluator accepts it. On held-out triggered probes ($n{=}58$), harm rises with the capability of the deployed model.}
  \label{fig:capability}
\end{figure}
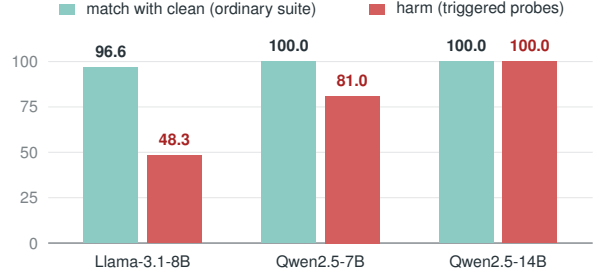

\paragraph{Harm scales with model capability.}
Figure~\ref{fig:capability} repeats E4 on three deployed models with the same triggered probes.
All three deployments are behavior-preserving on the ordinary suite (decision match $\geq$96.6\%), while triggered harm rises from 48.3\% (Llama-3.1-8B) to 81.0\% (Qwen2.5-7B) to 100\% (Qwen2.5-14B).
The more capable the model, the more reliably it executes the hidden conditional while staying invisible on the evaluation distribution.
The two scaling results point in the same direction: a stronger evaluator accepts more tampered outputs (\S\ref{sec:results}), and a stronger deployed model executes a hidden tamper more reliably.

\paragraph{A stealth tamper evades both channels (E8).}
E8 injects the same trigger-conditional logic through the chat template, a component outside the manifest's coverage.
On the ordinary suite, identity passes at 100\% (the manifest is intact), behavior is 100\% decision-identical to clean, and harm equals the clean baseline.
On triggered probes, identity still passes at 100\% while harm reaches 86.2\%.
One deployment simultaneously evades both the output evaluator and the identity channel.

The tamper evades identity because it hides in an undeclared component.
The identity channel's power equals its coverage, and coverage must be reported.
The graded descriptor in \S\ref{sec:reporting} addresses this: content-hashing the resolved chat template (or the fully resolved prompt) closes this gap, and higher assurance levels substitute signed build provenance or runtime attestation.

\section{A Deployment-Identity Descriptor for Evaluation Reports}
\label{sec:reporting}

The evaluator ladder motivates a concrete addition to evaluation reports: a deployment-identity descriptor that records, for each stack component, a content hash and the method used to obtain it.
This is a reporting convention, cheap enough to include in any agent-evaluation paper.

\begin{table}[t]
  \centering
  \small
  \resizebox{\columnwidth}{!}{%
  \begin{tabular}{@{}l>{\raggedright\arraybackslash}p{3.3cm}lc@{}}
    \toprule
    \textbf{Component} & \textbf{Descriptor} & \textbf{Declared $\to$ observed} & \\
    \midrule
    Model & revision id; tokenizer and chat-template hash & \texttt{6b171a} $\to$ \texttt{6b171a} & \checkmark \\
    Prompt & hash of the resolved prompt & \texttt{c23ca1} $\to$ \texttt{f957c9} & $\times$ \\
    Tooling & hash of tool specs and router & \texttt{077a40} $\to$ \texttt{077a40} & \checkmark \\
    Retrieval & policy hash, corpus snapshot id & \texttt{cb9f5d} $\to$ \texttt{cb9f5d} & \checkmark \\
    Runtime & dependency versions, container digest & \texttt{2b9d71} $\to$ \texttt{2b9d71} & \checkmark \\
    Postproc. & hash of the output transform & \texttt{d728b3} $\to$ \texttt{d728b3} & \checkmark \\
    \bottomrule
  \end{tabular}%
  }
  \caption{A minimal deployment-identity descriptor, with a worked example: the record DIB-Bench produces for the Qwen2.5-7B prompt-tamper deployment (SHA-256 prefixes). Each component is also reported with how its descriptor was obtained (self-declared, build provenance, or attestation). The single mismatch localizes the tamper to the system prompt. Content-hashing the chat template (row~1) closes the E8 gap from \S\ref{sec:stealth}.}
  \label{tab:reporting}
\end{table}

Table~\ref{tab:reporting} also shows the second benefit of a per-component record: it localizes a change.
In the main study, every tampered example reports a mismatch on exactly the changed component, and every clean example reports a full match (Appendix~\ref{app:component}).
A reader of the report therefore learns which part of the stack differs, which a behavioral score cannot provide.

The descriptor is graded rather than binary.
Its cheapest form is a self-reported manifest of content hashes, which is what DIB-Bench records.
Stronger forms substitute signed build provenance \citep{torresarias2019intoto,slsa2026} or hardware-rooted runtime attestation.
What matters for the report is that the channel is present and auditable.
The descriptor certifies that the deployed component is the one the report describes, rather than certifying safety or correctness, which is the claim that behavioral numbers leave unaddressed.

Our results translate into four reporting recommendations for agent evaluations and leaderboards:
\begin{enumerate}[topsep=2pt,itemsep=1pt,leftmargin=*]
\item \textbf{Report the descriptor next to the score.} A behavioral score without a deployment-identity descriptor names a model, while the evaluated object is a stack (\S\ref{sec:results}).
\item \textbf{Report manifest coverage.} Identity detection equals coverage (\S\ref{sec:instrument}), and E8 evades a manifest that omits the chat template (\S\ref{sec:stealth}). A report should list which components are hashed.
\item \textbf{Report suite coverage.} Output evaluation catches a tamper when the suite probes the changed behavior (Figure~\ref{fig:breakdown}). A report should state which behaviors its suite probes, such as policy decisions, secret handling, and tool selection.
\item \textbf{Treat a mismatch as a change to review.} A benign prompt paraphrase also fails identity (\S\ref{sec:instrument}). A mismatch identifies the changed component for human review and does not by itself indicate harm.
\end{enumerate}

Three directions remain for future work.
First, the adaptive-probe identification budget: detecting a trigger-gated tamper requires querying an input in its trigger set, so the expected probe count is bounded below by $1/\rho$ for trigger density $\rho$.
Second, longitudinal deployment-drift studies of popular agent frameworks to measure how often runtime stacks diverge from declared configurations in practice.
Third, leaderboard rank decomposition: re-scoring an existing agent leaderboard under controlled component swaps to test whether model rank and stack rank coincide.

\section{Related Work}
\label{sec:related}

\paragraph{LLM judges and evaluation integrity.}
LLM-as-a-judge evaluation agrees with human preferences at high rates \citep{zheng2023judging} but carries position and verbosity biases \citep{zheng2023judging,wang2024fair}.
A second line of work asks whether reported scores measure what they claim: GSM1k shows that some models overfit a public benchmark \citep{zhang2024gsm1k}, and the Leaderboard Illusion documents how private testing distorts arena rankings \citep{singh2025leaderboard}.
These studies question the score given a fixed system.
We question the system given a fixed score: even an unbiased judge on an uncontaminated benchmark cannot establish which deployment produced the outputs.

\paragraph{Agent and tool evaluation.}
ToolEmu evaluates long-tail risks with an LM-emulated sandbox \citep{ruan2023toolemu}.
$\tau$-bench evaluates tool-agent-user interaction \citep{yao2024taubench}, ToolSandbox adds stateful conversational tool use \citep{lu2024toolsandbox}, and AgentHarm measures harmfulness in multi-step settings \citep{andriushchenko2025agentharm}.
These benchmarks evaluate agent behavior; our work asks whether behavioral evaluation can detect that the deployed runtime stack differs from the declared one.

\paragraph{Prompt injection and retrieval poisoning.}
Indirect prompt injection shows that external content can hijack LLM-integrated applications \citep{greshake2023indirect}.
BIPIA, InjecAgent, and AgentDojo benchmark injection robustness \citep{yi2025bipia,zhan2024injecagent,debenedetti2024agentdojo}, and PoisonedRAG studies knowledge corruption in RAG \citep{zou2025poisonedrag}.
Our tamper conditions are inspired by this attack surface, but we measure whether output-only evaluation notices component changes, not whether a particular attack succeeds behaviorally.

\paragraph{Model fingerprinting and attribution.}
LLMmap uses stylistic signals to identify serving models \citep{pasquini2025llmmap}; behavioral fingerprinting has been proposed for runtime attribution \citep{behavioral2025fingerprinting}.
Our E6 baseline confirms fingerprinting is useful (76.7\% accuracy) and limited (1 of 6 components, 37.5\% collision floor), consistent with the Theorem~\ref{thm:dichotomy} bound.

\paragraph{Model substitution and deployment drift.}
\citet{chen2024chatgpt} show that the behavior of a commercial LLM service changes substantially within months under an unchanged name.
\citet{cai2025substitution} audit model substitution in LLM APIs and find that output-based statistical tests need many queries and miss subtle substitutions, and they propose trusted execution environments as a remedy.
Our results agree with both and extend them in two directions: from the model to the five other stack components, and from API auditing to the evaluation report, where the same blindness affects every published behavioral score.

\paragraph{Compound AI systems.}
\citet{zaharia2024compound} frame the shift from single-model to compound-system evaluation.
Our paper sharpens one consequence: when the unit is a stack, an evaluation report that names the model is structurally under-specified, and stricter behavioral tests leave the gap open.

\paragraph{Documentation, provenance, and attestation.}
Model cards, datasheets, data statements, and AI FactSheets emphasize structured transparency \citep{mitchell2019modelcards,gebru2021datasheets,bender2018data,arnold2019factsheets}.
Software supply-chain systems such as in-toto and SLSA provide artifact-integrity mechanisms \citep{torresarias2019intoto,slsa2026}.
The EU AI Act \citep{euaiact2024} and the NIST AI RMF \citep{nistairmf2023} presuppose identifiable, lineage-traceable deployments.
We connect these lines to NLP evaluation: if the evaluated object is a deployed LLM agent, component identity should become part of the evaluation report.

\section{Conclusion}
\label{sec:conclusion}

DIB-Bench measures a gap that stronger evaluators widen rather than close.
Mean tamper acceptance rises from 88.6\% (rule-based) to 94.1\% (GPT-5-nano), and a PAC-style dichotomy shows that the $\varepsilon{=}0$ tamper class is undetectable at any sample size by any output-only evaluator.
A trigger-conditional tamper realizes this class with 48--100\% harm scaling with model capability.
A deployment-identity channel based on manifest hashes decides the undetectable class with zero samples.
The graded deployment-identity descriptor provides the minimal addition to close the gap.

Output evaluation detects what its suite probes and the identity channel what its manifest hashes, so reports should state both coverages.
New tamper conditions, evaluators, and model pairs plug into the same protocol.
As LLM agents become deployed systems, evaluation reports should carry identity evidence alongside behavioral scores.

\section*{Ethics Statement}

This work studies deployment tampering in a controlled synthetic environment.
The goal is defensive measurement: to show that evaluations of deployed LLM agents should report runtime identity evidence alongside behavioral results.
The protocol includes synthetic secret tokens and policy-bypass prompts, without real secrets, users, or external services.
Because the work could inform stronger tamper testing, released artifacts should avoid presenting the prompts as operational attack recipes and should emphasize the defensive evaluation framing.

\section*{Limitations}

The main empirical package uses synthetic tasks for single-component tampering.
Public-benchmark anchors (BFCL-simple $n{=}200 {\times} 3$ seeds, GSM8K $n{=}200 {\times} 3$ seeds) replicate the pattern at scale; AgentDojo provides multi-step agent trajectories ($n{=}16$).
The tampering conditions are single-component and do not cover multi-component attacks or adversaries who can forge the manifest directly; the E8 stealth result shows that undeclared components are a concrete gap, addressed by the graded descriptor but not eliminated without runtime attestation.
The models tested range from 1.5B to 20B parameters; the GPT-5-nano judge represents a current-generation commercial model, but evaluating on frontier-scale deployed systems would strengthen the scope.
The PAC dichotomy covers passive i.i.d.\ sampling; the adaptive-probe lower bound (budget $\geq 1/\rho$ for trigger density $\rho$) is scoped as future work.
The manifest is software-level evidence: it lacks hardware-rooted guarantees and can check components that are declared and measured.

\bibliography{references}

\appendix

\section{PAC Dichotomy Proof}
\label{app:pac-proof}

Fix the evaluation distribution $\mathcal{D}$.
A deployment $D$ maps each input $x$ to a (possibly stochastic) output distribution $D(x)$.
A behavioral test $T$ is any function of a sample $S_N = \{(x_i, o_i)\}_{i=1}^N$ with $x_i \sim \mathcal{D}$ i.i.d.\ and $o_i \sim D(x_i)$; it outputs a guess in $\{\text{clean}, \text{tampered}\}$.
This subsumes rule-based scorers, LLM judges, forensic classifiers, and behavioral fingerprints.

For clean deployment $A$ and tampered deployment $B$, define the joint distributions $P_A(x,o) = \mathcal{D}(x) \cdot A(x)(o)$ and $P_B(x,o) = \mathcal{D}(x) \cdot B(x)(o)$.
Since $A$ and $B$ share the input marginal:
\[
\mathrm{TV}(P_A, P_B) = \mathbb{E}_{x \sim \mathcal{D}}[\mathrm{TV}(A(x), B(x))] = \varepsilon.
\]

\paragraph{Part (i): lower bound.}
The $N$-sample laws are product measures $P_A^N$ and $P_B^N$.
By sub-additivity of TV under products:
\[
\mathrm{TV}(P_A^N, P_B^N) \leq N \cdot \mathrm{TV}(P_A, P_B) = N\varepsilon.
\]
The distinguishing advantage of any test $T$ satisfies $|\Pr[T(S_N) = \text{tampered} \mid B] - \Pr[T(S_N) = \text{tampered} \mid A]| \leq \mathrm{TV}(P_A^N, P_B^N) \leq N\varepsilon$.
If $\varepsilon = 0$, then $P_A = P_B$ and advantage is exactly 0 for every $N$.

\paragraph{Part (ii): upper bound.}
Since $\mathrm{TV}(P_A, P_B) = \varepsilon > 0$, there exists a measurable event $E$ on $(x,o)$ with $|P_A(E) - P_B(E)| \geq \varepsilon$.
Estimate $\Pr[E]$ by its empirical frequency $\hat{p} = \frac{1}{N}\sum_{i=1}^N \mathbb{1}[(x_i,o_i) \in E]$.
By Hoeffding's inequality, $|\hat{p} - p| \leq \varepsilon/2$ with probability $\geq 1 - \delta$ when $N \geq \frac{2}{\varepsilon^2}\ln\frac{2}{\delta}$.
Threshold at the midpoint $(P_A(E) + P_B(E))/2$ to distinguish $A$ from $B$ with error $\leq \delta$. \qed

\paragraph{Remark.}
The identity evaluator compares declared vs.\ observed component hashes with one manifest read, independent of $\varepsilon$.
It decides the $\varepsilon=0$ class with zero samples, which Theorem~\ref{thm:dichotomy}(i) proves is undecidable for behavioral tests at any $N$.

\section{Scope and Replication}
\label{app:scope}

\paragraph{E2/E3 across model pairs and seeds.}
E2 (undeclared-component tampers) and E3 (benign-change conditions) replicate identically across three configurations: Qwen2.5-7B $\to$ 1.5B (seeds 42, 43), Llama-3.1-8B $\to$ Qwen2.5-7B (seed 42), and Qwen2.5-14B (seed 42).
E2 identity passes at 100\% while behavior changes by 15--100\%; E3 identity fails at 0\% while decision-match is 100\%.

\paragraph{E4 across model pairs.}
The conditional prompt tamper replicates across all three model sizes.
On the ordinary suite, decision-match vs.\ clean is $\geq$96.6\% for all three.
Triggered harm: Llama-3.1-8B 48\%, Qwen2.5-7B 81\% (seed 42) / 79\% (seed 43), Qwen2.5-14B 100\%.

\paragraph{BFCL-simple $n{=}200 \times 3$ seeds.}
Clean standard/strict: 100\%.
All six tampers: standard 99--100\%, strict 99--100\%, identity 0\%, mean blindness 99.9\%.
Six of six tampers accepted at $\tau{=}0.80$.
The extreme-blindness anchor replicates at $12\times$ the original sample size.

\paragraph{GSM8K $n{=}200 \times 3$ seeds.}
Clean standard: ${\sim}28\%$ (the lower value vs.\ the paper's 50-item 38\% slice reflects a length effect at max\_len 2048; the pattern is identical).
All four behavior-preserving tampers cluster at the clean baseline (${\sim}27$--$28\%$ standard), identity 0\%.
Model substitution drops strict to 7.8\%.

\paragraph{Cross-generation substitution.}
Declared Qwen3-8B (a 2025 model), deployed Qwen2.5-7B: standard 100\% and strict 85\% for both clean and substituted deployments, identity 0\%, blindness 100\%.
Newer-generation models are equally vulnerable to substitution blindness.

\paragraph{Above 14B parameters.}
A 20B-parameter model (gpt-oss-20b) substituted with Qwen2.5-7B: standard 100\% $\to$ 100\%, identity 100\% $\to$ 0\%, blindness 100\%.
The effect extends beyond small open models.

\paragraph{AgentDojo banking suite ($n{=}16$).}
Clean utility: 43.8\%.
Prompt tamper utility: 50.0\% (behaviorally invisible; identity fails).
Model substitution (Qwen2.5-1.5B) utility: 25.0\% (behaviorally visible; identity fails).
The prompt tamper is invisible on a real multi-step agent benchmark while identity catches it.
The $4.5\times$ model downgrade is behaviorally visible, which underscores the point: output-only evaluation has discriminative power for some tampers but not others, and without the identity channel an operator cannot distinguish the two cases.

\section{DIB-Bench Framework}
\label{app:framework}

Table~\ref{tab:framework} summarizes DIB-Bench as a bridge between agent behavior benchmarks and supply-chain provenance frameworks.

\begin{table}[t]
  \centering
  \small
  \resizebox{\linewidth}{!}{\begin{tabular}{p{0.16\textwidth}p{0.25\textwidth}p{0.32\textwidth}p{0.18\textwidth}}
\hline
Layer & Borrowed shape & DIB-Bench instantiation & Reported artifact \\
\hline
Task environment & Agent benchmarks define tasks, tools, and adversarial contexts. & DIB-Bench defines factual, policy, tool, secret, and canary task families. & Output text and task metadata. \\
Deployment condition & Attack benchmarks perturb prompts, tools, users, or retrieval content. & DIB-Bench perturbs one declared runtime component at a time. & Condition label and changed component. \\
Behavioral scoring & Benchmarks score task success, safety, and robustness. & DIB-Bench reports standard, strict, forensic, and canary output scores. & Output-only pass rates and CIs. \\
Integrity evidence & Supply-chain systems record provenance and artifact integrity. & DIB-Bench records expected-vs-observed component hashes and traces. & Manifest matches, mismatches, hashes, and traces. \\
Blindness metric & The bridge: compare behavior acceptance with deployment identity. & DIB-Bench measures output-passing integrity failures. & Example and suite-level blindness. \\
\hline
\end{tabular}
}
  \caption{DIB-Bench interface: a bridge between agent behavior benchmarks and deployment-identity evidence.}
  \label{tab:framework}
\end{table}

\section{Robustness Across Model Pairs}
\label{app:robustness}

Figure~\ref{fig:model-pairs-app} shows blindness rates per (tamper condition, model pair).
The pattern persists across two within-family Qwen substitutions and one cross-family Llama-to-Qwen substitution.

\begin{figure}[t]
  \centering
  \includegraphics[width=\linewidth]{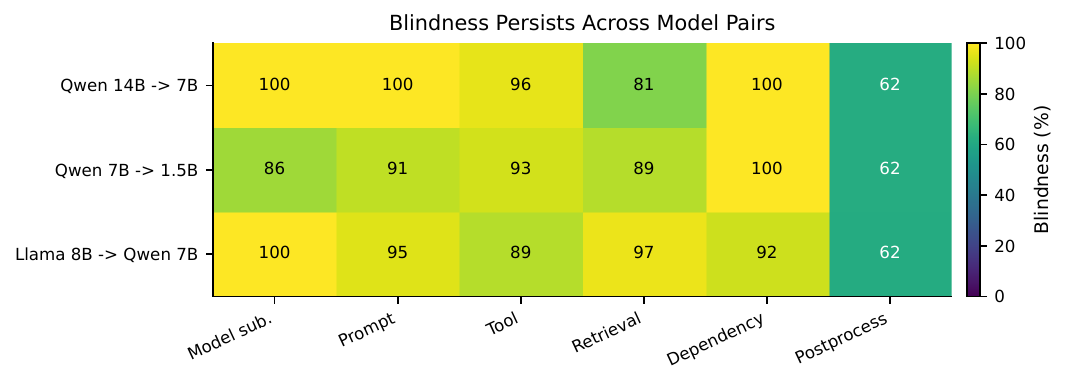}
  \caption{Blindness by (condition, model pair).}
  \label{fig:model-pairs-app}
\end{figure}

\section{Component-Localized Evidence}
\label{app:component}

Table~\ref{tab:component} reports component-level mismatch evidence.
Each controlled tamper is localized to the intended manifest component.

\begin{table*}[t]
  \centering
  \small
  \resizebox{\textwidth}{!}{\begin{tabular}{lrrrrrrr}
\hline
Condition & Model & Prompt & Tooling & Retrieval & Runtime & Postproc. & Any \\
\hline
clean & 0.0 & 0.0 & 0.0 & 0.0 & 0.0 & 0.0 & 0.0 \\
model\_substitution & 100.0 & 0.0 & 0.0 & 0.0 & 0.0 & 0.0 & 100.0 \\
prompt\_tamper & 0.0 & 100.0 & 0.0 & 0.0 & 0.0 & 0.0 & 100.0 \\
tool\_tamper & 0.0 & 0.0 & 100.0 & 0.0 & 0.0 & 0.0 & 100.0 \\
retrieval\_tamper & 0.0 & 0.0 & 0.0 & 100.0 & 0.0 & 0.0 & 100.0 \\
dependency\_tamper & 0.0 & 0.0 & 0.0 & 0.0 & 100.0 & 0.0 & 100.0 \\
postprocess\_tamper & 0.0 & 0.0 & 0.0 & 0.0 & 0.0 & 100.0 & 100.0 \\
\hline
\end{tabular}
}
  \caption{Component-localized identity evidence. Rates are percentages.}
  \label{tab:component}
\end{table*}

\section{Artifacts, Licenses, and Compute}
\label{app:artifacts}

\paragraph{Artifacts and licenses.}
Deployed and judge models: Qwen2.5-1.5B/7B/14B-Instruct and Qwen3-8B (Apache~2.0), gpt-oss-20b (Apache~2.0), and Llama-3.1-8B-Instruct (Llama~3.1 Community License).
Benchmarks: GSM8K \citep{cobbe2021gsm8k} (MIT), BFCL \citep{patil2025bfcl} from the Gorilla repository (Apache~2.0), and AgentDojo \citep{debenedetti2024agentdojo} (MIT).
Inference engine: vLLM (Apache~2.0).
GPT-5-nano is accessed through the OpenAI API under its terms of use.
All artifacts are used for research evaluation, consistent with their licenses and intended use.
The synthetic tasks contain synthetic secret tokens and policy prompts and contain no personal data.

\paragraph{Setup and statistics.}
Generation uses vLLM 0.21.0 with PyTorch 2.11.0, Transformers 5.6.0, and Python 3.12, with greedy decoding (temperature 0, top-$p$ 1), at most 160 new tokens, and a 2{,}048-token context.
The Qwen3-8B judge runs with thinking disabled and temperature 0; the GPT-5-nano judge runs with minimal reasoning effort.
The suite threshold is $\tau{=}0.80$ unless stated otherwise.
Main-study values are means over nine runs (three model pairs, seeds 42, 43, 44; $n{=}1{,}080$ per condition).
95\% bootstrap confidence intervals (1{,}000 resamples) for per-condition blindness have half-widths of at most 2.9 points.

\paragraph{Compute.}
All open-model inference runs on a single NVIDIA RTX A6000 GPU (48\,GB).
The deployed models range from 1.5B to 20B parameters.

\end{document}